\long\global\def\C#1\F{{}}
\documentclass[12pt]{article}
\pdfoutput=1

\usepackage{amsmath, amsthm, amssymb, graphicx, verbatim, eucal, color, multicol,
makeidx}
\usepackage{multirow}
\usepackage{pdfsync}
\usepackage{ footmisc} 
\usepackage{tikz,pgflibraryarrows,pgflibrarysnakes,pgflibraryshapes} 
\usepackage[all]{xy}
\usepackage{listings}
\usepackage[round]{natbib}
\swapnumbers

\newtheorem{theorem}{Theorem}
\newtheorem{lemma}[theorem]{Lemma}
\newtheorem{definition}[theorem]{Definition}
\newtheorem{remark}[theorem]{Remark}
\newtheorem{remarks}[theorem]{Remarks}
\newtheorem{example}[theorem]{Example}
\newtheorem{examples}[theorem]{Examples}
\newtheorem{fourexamples}[theorem]{Four Examples and One Counterexample}

\newtheorem{gauss graph}[theorem]{The Graph of Gauss's Puzzle}
\newtheorem{arrang hyper}[theorem]{An Arrangement of Hyperplanes}

\newtheorem{note}[theorem]{Note}
\newtheorem{notes}[theorem]{Notes}

\newtheorem{returnpytha}[theorem]{Application to the Pythagorean Theorem}

\let\slantedreturnpytha\returnpytha
\def\returnpytha{\slantedreturnpytha\rm}
\newtheorem{orderinvariance}[theorem]{Order-invariance}
\let\slantedorderinvariance\orderinvariance
\def\orderinvariance{\slantedorderinvariance\rm}
\newtheorem{openproblem}[theorem]{Open Problem}
\let\slantedopenproblem\openproblem
\def\openproblem{\slantedopenproblem\rm}

\renewenvironment{proof}{{\sc Proof.}}{\EOP\wl}

\def\tl{\vskip 2mm}
\def\wl{\vskip 4mm}

\def\][{\hspace{.03cm}]\hspace{-.1cm}[}

\def\SB{\subseteq}

\def\hh{\theta}

\def\Ree{\mathbb R}
\def\Re+{\mathbb R_+}
\def\Re++{\mathbb R_{++}}
\def\Na{\mathbb N}

\def\EQ{\Longleftrightarrow}
\def\EQQ{\quad\Longleftrightarrow\quad }

\def\EOP{\phantom{a}\hfill $\square$}

\def\too{\longrightarrow}

\def\XXX{{\cal X}}

\def\begeq{\begin{equation}}
\def\edeq{\end{equation}}
\def\st{{\;\vrule height8pt width0.7pt depth2.5pt\;}}

\def\sqr#1#2{{\vcenter{\vbox{\hrule height.#2pt
          \hbox{\vrule width.#2pt height#1pt \kern#1pt
              \vrule width.#2pt}
           \hrule height.#2pt}}}}
\def\square{\mathchoice\sqr56\sqr56\sqr{2.1}3\sqr{1.5}3}
\def\roster{\begin{enumerate}}
\def\endroster{\end{enumerate}}

\def\fp{\noindent}

\reversemarginpar

\let\slantedexample\example
\def\example{\slantedexample\rm}
\let\slantedexamples\examples
\def\examples{\slantedexamples\rm}

\let\slantedremark\remark
\def\remark{\slantedremark\rm}
\let\slantedremarks\remarks
\def\remarks{\slantedremarks\rm}
\let\slanteddefinition\definition
\def\definition{\slanteddefinition\rm}
\let\slantednote\note
\def\note{\slantednote\rm}
\let\slantednotes\notes
\def\notes{\slantednotes\rm}

\makeatletter
\long\def\@makecaption#1#2{%
  \vskip\abovecaptionskip
  \sbox\@tempboxa{\small #1: \sc #2}%
  \ifdim \wd\@tempboxa >\hsize
    \small #1: \sc #2\par
  \else
    \global \@minipagefalse
    \hb@xt@\hsize{\hfil\box\@tempboxa\hfil}%
  \fi
  \vskip\belowcaptionskip}
\makeatother

\begin{document}

\newcommand{\fl}{\marginpar{\hfill$\rightarrow$}}
\normalsize
\renewcommand{\baselinestretch}{1}

\pagestyle{plain}
\pagestyle{plain}
\def\qqquad{\qquad\quad}

\title{
\vspace*{-1.8cm}
\large On a bounded version of H\"older's Theorem\\ 
and its application to  the permutability equation}
\author{\normalsize Jean-Claude Falmagne\\
{\normalsize University of California, Irvine}
}
\date{\small \today}
\maketitle
\def\too{\longrightarrow}
\thispagestyle{empty}
\vspace{-1cm}

\begin{quote}
{\sl This chapter is dedicated to Patrick Suppes, whose works and counsel have shaped much of my scientific life.} 
\end{quote}

\begin{abstract}
\tl
\fp
{\small The permutability equation $G(G(x,y),z) = G(G(x,z),y)$ is satisfied by many scientific and geometric laws. A few examples among many are: The Lorentz-FitzGerald Contraction, Beer's Law, the Pythagorean Theorem, and the formula for computing the volume of a cylinder. We prove here a representation theorem for the permutability equation, which generalizes a well-known result.  
The proof is based on a bounded version of H\"older's Theorem.}
\end{abstract}
\tl
\fp
Holder's Theorem on ordered groups  is a foundation stone of measurement theory \citep[c.f.][]{krantzetal71, suppesetal89, luceetal90}, and so, of much of quantitative science. There are several renditions of it. Whatever the version, the theorem concerns an algebraic structure $(\XXX, \circ, \precsim)$, in which $\XXX$ is a set, $\circ$ is an operation on $\XXX$, and\, $\precsim$\, is a weak order on $\XXX$ (transitive, connected), which may be a simple order (antisymmetric). The axioms imply the existence of a function $f: \XXX\to\Ree$ such~that 
\begin{align*}
&x\precsim y \quad\EQ\quad f(x)\leq f(y)&& \\[2mm]
&\,\,\,f(x\circ y) = f(x) + f(y)&&(\text{whenever $x\circ y$ is defined}).
\end{align*}
Most formulations of this theorem have one or both of two drawbacks. 
\begin{roster}
\item[] {\sc Hypothesis 1.} The elements of $\XXX$ can be arbitrarily large. 
\item[] {\sc Hypothesis 2.} The elements of $\XXX$ can be arbitrarily small.
\end{roster}
From the standpoint of social sciences applications, both of these hypotheses are unwarranted because the sensory mechanisms of humans and animals restrict the range of usable stimuli. In psychophysics, for example, small stimuli are undetectable by the sensory mechanisms, and large ones would damage them. Even in physics (relativity) the hypothesis that infinitely large quantities exist is inconsistent with current theories. In the axiomatization of \citet{lucemarley69} \citep[see also][page 84]{krantzetal71}, arbitrarily large elements need not exist. However, they use the following solvability axiom, which essentially asserts the existence of arbitrarily small elements. 
\begin{roster}
\item[]{\sl If $x\prec y$, then there is some $z$ such that $x\circ z\precsim y$.}
\end{roster} 

It might be argued that these two hypotheses are idealizations, and that using them simplifies the derivations.
The trouble is that, in the framework of the other axioms,  these two hypotheses imply that the operation $\circ$  is commutative. But commutativity is an essential property, which is testable empirically. To derive such a property from questionable axioms is not ideal.  Our Lemma \ref{Holder} is a version of H\"older's Theorem, due to \citet{falma75a},  in which neither arbitrarily small, nor arbitrarily large elements are assumed to exist, and in which commutativity is an independent axiom. 
\tl  

We use this lemma prove a representation theorem for the `permutability' property, which is 
an abstract constraint on a real, positive valued function $G$ of two real positive variables.  This property is formalized by the equation
\begin{equation}\label{permut_basic}
G(G(y,r),t) = G(G(y,t),r),
\end{equation}
where $G$ is strictly monotonic and continuous in both  variables. An  interpretation of $G(y,r)$ in Equation (\ref{permut_basic}) is that the second variable $r$ in modifies the state of the first variable $y$, creating an effect evaluated by $G(y,r)$ in the same measurement variable as~$y$. The left hand side of (\ref{permut_basic}) represents a one-step iteration of this phenomenon, in that $G(y,r)$ is then modified by $t$, resulting in the effect $G(G(y,r),t) $. Equation (\ref{permut_basic}), which is referred to as the `permutability' condition in the functional equations literature  \citep[c.f.][]{aczel:66}, formalizes the concept that the order of the two modifiers $r$ and $t$ is irrelevant. The importance of that property for scientific applications is that it can sometimes be inferred from a {\sl gedanken} experiment, before any experimentation, thereby substantially  constraining the possible models for a situation. 

 Indeed, under fairly general conditions of continuity and solvability making empirical sense, the permutability condition (\ref{permut_basic}) implies the existence of a general representation
\begin{equation}\label{basic_representation}
G(y,r)=f^{-1}(f(y) + g(r)),
\end{equation} 
where $f$ and $g$ are real valued, strictly monotonic continuous functions. We prove this fact here in the form of our Theorem \ref{permut => Holder},  generalizing results of  \citet{hossz62a, hossz62b, hossz62c} \citep[cf.~also][]{aczel:66}. It is easily shown that the representation (\ref{basic_representation}) implies the permutability condition 
(\ref{permut_basic}): we have
\begin{align*}
G(G(y,r),t)&=f^{-1}(f^{-1}(f(G(y,r))+g(t)))&&(\text{by (\ref{basic_representation})})\\
&=f^{-1}(f^{-1}(f(f^{-1}(f(y) + g(r)))+g(t)))&&(\text{by (\ref{basic_representation}) again})\\
&=f^{-1}(f(y) + g(r) + g(t))&&(\text{simplifying})\\
&=f^{-1}(f(y) + g(t) + g(r))&&(\text{by commutativity})\\
&=G(G(y,t),r)&&(\text{by symmetry})\,.
\end{align*}
We will also use a more general condition, called `quasi permutability', which is defined by the equation
\begin{align}\label{quasi_perm_basic}
M(G(y,r),t) &= M(G(y,t),r)\\
\noalign{\hspace{-.55cm}and lead to the representation}
\label{quasi_perm_repres}
M(y,r) &=m((f(y) + g(r))\,.
\end{align}

In our first section, we state some basic definitions and we 
describe a few examples of laws, taken from physics and geometry, in which the permutability condition  applies. We also give one example, van der Waals Equation, which is not permutable. The second section is devoted to  some preparatory lemmas.  The last  section contain the main results of the paper.
\section*{Basic Concepts and Examples}

\begin{definition}\label{basic defs} We write $\Ree_+$ and $\Ree_{++}$ for the nonnegative and the positive reals, respectively.  Let $J$, $J' $, and $H$ be  real nonempty and nonnegative intervals. A \emph{(numerical) code} is a function $M: J\times J'  \,\, \overset{\text{\tiny onto}}{\too}\,\, H$ which is strictly increasing in the first variable, strictly monotonic in the second one, and continuous in both. A code $M$ is \emph{solvable} if it satisfies the following two conditions.
\begin{roster}
\item[{[S1]}] If $M(x,t) < p\in H$, there exists $w\in J$ such that $M(w,t) = p$. 
\item[{[S2]}]  The function $M$ is \emph{1-point right solvable}, that is, there exists a point $x_0\in J$ such that for every $p\in H$, there is $v\in J'$ satisfying $M(x_0,v)= p$. In such a case, we may say that  $M$ is \emph{$x_0$-solvable}.
\end{roster}
By the strict monotonicity of $M$, the points $w$ and $v$ of [S1] and [S2] are unique. 
\tl

 Two functions $M: J\times J'  \to H$ and $G :  J\times J'  \to H'$ are \emph{comonotonic} if 
\begin{gather}\label{comonotonic eq 1}
M(x,s)\leq M(y,t)\EQQ  G(x,s)\leq G(y,t),\qquad\quad(x,y\in J; s,t\in J' ) .\\[2mm]
\noalign{In such a case, the equation\vspace{.4mm} }
\label{comonotonic eq 2}
F(M(x,s)) = G(x,s)\quad\qquad\quad\qquad(x\in J; s\in J' )
\end{gather}
defines a strictly increasing continuous function 
$F:H \,\, \overset{\text{\tiny onto}}{\too}\,\,H'$.  We may say then that $G$ is \emph{$F$-comonotonic} with $M$.
\end{definition}

We turn to the key condition of this paper. 
\begin{definition}
A function $M: J\times J'  \too H$ is \emph{quasi permutable} if there exists a function $G: J\times J'  \to J$ co-monotonic with $M$  such that 
\begin{equation}\label{quasi permut}
M(G(x,s),t) = M(G(x,t),s)\qquad\quad(x,y\in J; s,t\in J' ). 
\end{equation}

We say in such a case that $M$ is \emph{permutable with respect to} $G$, or  \emph{$G$-permutable} for short.  When $M$ is permutable with respect to itself, we simply say that $M$ is \emph{permutable}, a terminology consistent with \citet[][Chapter 6, p.~270]{aczel:66}.
\end{definition}

We mention the straightforward consequence:

\begin{lemma}\label{M permut => G permut}
A function  $M: J\times J'  \to H$ is $G$-permutable only if $G$ is permutable.
\end{lemma}
\begin{proof} Suppose that $G$ is $F$-comonotonic with $M$. For any $x\in J$ and  $s,t\in J' $, we get
$G(G(x,s),t) = F(M(G(x,s),t))= F(M(G(x,t),s))=G(G(x,t),s).$
\end{proof}
Many scientific laws embody  permutable or quasi permutable numerical codes, and hence can be written in the form of Equation (\ref{basic_representation}). We give four quite different examples below.  In each case, we derive the forms of the functions $f$ and $g$ in the representation equation (\ref{basic_representation}). 
\begin{fourexamples}\label{permut examples}{\rm \hfill\\[2mm]
(a) {\sc The Lorentz-FitzGerald Contraction.} This term denotes a phenomenon in special relativity, according to which the apparent length of a rod measured by an observer moving  at the speed $v$ with respect to that rod is a decreasing function of~$v$, vanishing as $v$ approaches the speed of light. This function is specified by the formula
\begin{equation}\label{standard LF}
L (\ell,v) =\ell \sqrt{
1-\left (\frac
v{c}\right )^2},
\end{equation}
in which $c>0$ denotes the speed of light,  $\ell$ is the actual length of the rod (for an observer at rest with respect to the rod), and 
$L :\Ree_+\times [0,c[\,\, \overset{\text{\tiny onto}}{\too}\,\,\Ree_+$ is the length of the rod measured by the moving observer.

The function  $L$ is a permutable code.  Indeed, $L$ satisfies the strict monotonicity and continuity requirements, and we have
\begin{align}\label{standard LF permut}
L(L (p,v),w) ={p} \left(1-\left (\frac v{c}\right )^2\right)^{-\frac 12} 
\left(1-\left (\frac w{c}\right )^2\right)^{-\frac 12}
= L(L(p,w),v).
\end{align} 
Solving the functional equation
\begin{equation}\label{fe_LF}
\ell \sqrt{1-\left(\frac vc\right)^2}= f^{-1}(f(\ell) +g(v))
\end{equation}
leads to the Pexider equation \citep[c.f.][pages 141-165]{aczel:66} 
\begin{align}\label{pexi_a}
f(\ell y) &= f(\ell ) + k(y)\\
\nonumber
\text{with}\qquad\quad  k(y) &=  g\left(c \sqrt{1-y^2}\,\right).
\end{align}
As the background conditions (monotonicity and domains of the functions\footnote{Note that the standard solutions for Pexider equations are valid when the domain of the equation is an open connected subset of $\Ree^2$ rather than $\Ree^2$ itself. Indeed,  \citet[][see also  Acz\'el, 2005, Chudziak and Tabor, 2008, and Rad\'o and Baker, 1987]{aczel:87}\nocite{aczel:05}\nocite{rado87}\nocite{chudziak08}
has shown that, in such cases,  this equation can be extended to the real plane.}) are satisfied,
the unique forms of $f$ and $g$ in (\ref{pexi_a}) are determined. They are: with $\xi > 0$,
\begin{align}\label{sol_f_LF}
f(\ell ) &= \xi\ln \ell  + \hh\\
\label{sol_g_LF}
g(v) &= \xi\ln \left (\sqrt{1-\left(\frac vc\right)^2}\right)\,.
\end{align}
\tl

(b) {\sc Beer's Law.} This law applies in a class of empirical situations where an incident radiation
traverses some absorbing medium, so that only a fraction of the 
radiation goes through. In our notation, the expression of the law is
\begin{equation}\label{standard BL}
I (x,y) =x \,e^{-\frac y{c}}, \qquad\qquad(x,y\in \Ree_+,\, c\in\Ree_{++}\text{ constant})
\end{equation}
in which $x$ denotes the intensity of the incident light, $y$ is the
concentration  of the absorbing medium, $c$ is a reference level, and
$I(x,y)$ is the intensity of the transmitted radiation. The form of this law is similar to that of the Lorentz-FitzGerald Contraction and the same arguments apply. Thus,  the function $I: \Ree_+\times\Ree_+\overset{\text{\tiny onto}} {\too}\Ree_+$ is also a permutable code. 
The solution of the functional equation 
\begin{align*}
x \,e^{-\frac y{c}}&= f^{-1}(f(x) + g(y))
\end{align*}
follows a pattern identical to that of Equation (\ref{fe_LF}) for the Lorentz-FitzGerald Contraction. The only difference lies in the definition of the function $g$, which is here
$$
g(y) = -\xi \frac yc\,.
$$
The definition of $f$ is the same, namely (\ref{sol_f_LF}).
So, we get
\begin{equation*}
I(x,y)= f^{-1}(f(x) + g(y))
= \exp\left( \frac 1\xi( \xi\ln x + \hh - \xi \frac yc- \hh\right)
= x \,e^{-\frac yc}\,.
\end{equation*}

(c) {\sc The volume of a cylinder.}  The permutability equation applies not only to many physical laws, but also to some fundamental formulas of geometry, such as the volume $C(\ell, r)$ of a cylinder of radius $r$ and height $\ell$, for example. In this case, we have
\begin{equation}\label{cylinder}
C(\ell, r) = \ell \pi r^2,
\end{equation}
which is permutable. We have
$$
C(C(\ell, r), v) = C(\ell \pi r^2, v) = \ell \pi r^2 \pi v^2 = C(C(\ell, v), r).
$$
Solving the functional equation
$$
 \ell \pi r^2 = f^{-1}(f(\ell) + g(r))
$$
yields the solution
\begin{align*}
f(\ell) & = \xi \ln \ell + \hh
\end{align*}
(again, the function $f$ is the same as in the two preceding examples), and
\begin{align*}
g(r)&= \xi \ln \left(\pi r^2\right)\,,
\end{align*}
with
$$
f^{-1}(f(\ell) + g(r)) = \exp\left(\frac 1\xi \left(\xi\ln \ell +\hh + \xi \ln \left(\pi r^2\right)-\hh\right)\right) = \ell \pi r^2\,. 
$$

We give another geometric example below, in which the form of $f$ is different.
\tl

(d) {\sc The Pythagorean Theorem.}  The function
\begin{equation}\label{Pytha}
\hspace{1cm}P(x,y) = \sqrt{x^2 + y^2}\qqquad\qqquad(x,y \in \Ree_{++}),
\end{equation}
representing the length of the hypothenuse of a right triangle in terms of the lengths of its sides, 
is a  permutable code.
We have indeed 
\begin{gather*}
P(P(x,y),z) = \sqrt {P(x,y)^2 + z^2}  
= \sqrt{x^2+y^2 + z^2} 
= P(P(x,z),y).
\end{gather*}
The function $P$ is symmetric. So we must solve the equation
\begin{align}\nonumber
\sqrt{x^2 + y^2} &= f^{-1}\left(f(x) + f(y)\right)\\
\noalign{or, equivalently,}
\label{Cauchy_pytha}
f\left(\sqrt{x^2+y^2}\right)&=f(x) +f(y)\,.
\end{align}
With $z=x^2$, $w=y^2$, and defining the function $h(z) = f\left(z^{\frac 12}\right)$, Equation (\ref{Cauchy_pytha}) becomes
$$
h(z+w)= h(z) + h(w)\,,
$$
a Cauchy equation on the positive reals, with $h$ strictly increasing. It has the  unique solution
$
h(z) = \xi \,z,
$ 
for some positive real number $\xi$ \citep[c.f.][page 31]{aczel:66}. We get 
\begin{align*}
f(x) &= \xi x^2\\
\noalign{ and}
f^{-1}(f(x) + f(y)) &= \left(\frac 1\xi\left(\xi x^2+ \xi y^2\right)\right)^{\frac 12} = \sqrt{x^2+y^2}\,.
\end{align*}

(f) {\sc The Counterexample: van der Waals Equation.} One form of this equation is
\begin{equation}\label{vdw_equation}
 T(p,v)= K\left(p+\frac a{v^2}\right)(v-b),
 \end{equation}
 in which $p$ is the pressure of a fluid, $v$ is the volume of the container, $T$ is the temperature, and $a$, $b$ and $K$ are constants; $K$ is the reciprocal of the Boltzmann constant. It is easily shown that the function $T$ in 
 (\ref{vdw_equation}) is not permutable.
}
\end{fourexamples}
\def\F1-1{F^{-1}_1}
\begin{openproblem}
Examining the four examples (a) to (d) above suggests that once the exact form of a permutable law is known, the form of the functions $f$ and $g$ in the representation (\ref{basic_representation}) can easily be guessed. For example, in each of the problems (a), (b), and (c), the permutable law is the product of two functions, with the first one being the identity function. In these problems, the form of $f$ is the same, namely $f(x) = \xi\ln x + \hh$. A more difficult problem is: are there basic structural properties which, in addition to  permutability,  determine the form of a permutable law, possibly up to some parameters? We will consider this problem in a later paper.
\end{openproblem}
\section*{Preparatory Results} 
The main step in our developments is based on the following construction.
\begin{definition}\label{def operation} Suppose that $G:J\times J'\to J$ is a code that is $x_0$-solvable in the sense of Condition [S2].
 Define the operation $\bullet$ on $J$ by the equivalence
\begin{equation}\label{def circ}
x\bullet y = G(x, v)\quad\EQ\quad G(x_0, v) = y\qquad\quad(x,y\in J;v\in J').
\end{equation}
\end{definition}
We show in this section that a solvable code $G$ is permutable if and only if it has an additive representation
\begin{equation}\label{additive G 1}
G(y,v) = f^{-1}(f(y) + g(v))\qqquad(x,y\in J;v\in J')
\end{equation}
where $f: J \to \Ree_+$ and $g: J'\to \Ree_+$ are continuous functions with $f$ strictly increasing and $g$ strictly monotonic.
\tl
Our basic tool lies in the following lemma. 

\begin{lemma}\label{Holder} Let $J$ be a real non degenerate interval. With $R\SB J\times J$, let\, 
$\bullet :R\to J$ be a non necessarily closed operation on $J$. We write $xRy$ to mean that $x\bullet y$ is defined. Suppose that the triple 
$(J, \bullet , \leq)$,  where $\leq$ is the inequality of the reals,  satisfies the following five independent  conditions:
\begin{roster}
\item[{\rm (i)}] $yRx$ if $xRy$, and when $yRx$, then $y\bullet x = x\bullet y$\,;
\item[{\rm (ii)}]  whenever $yR x$, $wR z$, $ wRy'$, $z'Rx $, $yR y'$ and $z'R z$, then
$$
(y\bullet x = w\bullet z)\,\, \text{and}\,\, (w\bullet y'=z'\bullet x) \quad\text{imply}\quad 
y\bullet y' = z'\bullet z\,;
$$
\item[{\rm (iii)}]  there exists $x\in J$ such that $xRx$ and $x\bullet xRx$\,;
\item[{\rm (iv)}]  if $y\bullet x < z$, then $y\bullet w = z$ for some $w$ in $J$\,;
\item[{\rm (v)}]   for every $x$, $y$ and $z$ in $J$, with $x<y$, the set $N(x,z;y) = \{n\in\Na^+\st  x_y^n\leq z\}$ is finite, where the sequence $(x_y^n)$ is defined recursively as follows:
\begin{roster}
\item[{\rm (a)}] $x_y^1 = x$;
\item[{\rm (b)}]  if $x_y^{n-1}$ is defined and $x'$ exists such that $y \bullet x_y^{n-1} = x\bullet x'$ then $x_y^n = x'$. 
\end{roster}
\end{roster}
Then, there exists a strictly increasing function $f:J\to J$ such that
$$
f(x\bullet y) =f(y) + f(y).
$$
\end{lemma}
 \cite[For a proof, see][]{falma75a}.
\begin{lemma} \label{G satisfies 1-5} Let $G:J\times J'\to J$ be a solvable, permutable code. Then, the triple $(J,\bullet, \leq )$, with the operation $\bullet$ defined by {\rm (\ref{def circ})}, satisfies  Conditions {\rm (i)-(v)} of Lemma~{\rm \ref{Holder}}. Moreover, the operation $\bullet$ is associative, strictly increasing and continuous in both variables.
\end{lemma}

\begin{proof} Take any $x,y\in J$ with 
\begin{align}\label{commut 1}
G(x_0,r) &= x\\
\noalign{and} 
\label{commut 2}
G(x_0, v) &= y.
\end{align}

(i) 
By (\ref{def circ}), (\ref{commut 1}), (\ref{commut 2}) and the permutability of $G$, we get successively, 
$$
y\bullet x = G(y,r)= G(G(x_0,v),r)=G(G(x_0,r),v) = G(x,v) = x\bullet y.
$$

(ii) Suppose that
\begin{equation}\label{simplifyability eq o}
(y\bullet x = w\bullet z)\,\, \text{and}\,\, (w\bullet y'=z'\bullet x) .
\end{equation}
With (\ref{commut 1}),  (\ref{commut 2}) and 
\begin{equation}\label{def the G(x0)'}
G(x_0, s) = z, \,\,
G(x_0, t) = w, \,\,G(x_0, v') = y', \,\,G(x_0, s') = z',
\end{equation}
we get from (\ref{simplifyability eq o})
\begin{align}\label{simplifyability eq G1}
G(y,r)&= G(w,s)\\
\label{simplifyability eq G2}
G(w,v') &= G(z',r)\,.
\end{align}
Equation (\ref{simplifyability eq G1}) gives
$$
G(G(y,r),v') = G(G(w,s), v'),
$$
which yields successively
\begin{align*}
G(G(y,v'),r) &= G(G(w,v'), s)&&(\text{by permutability})\\
&=G(G(z',r), s)&&(\text{by (\ref{simplifyability eq G2})})
\\
&=G(G(z',s), r)&&(\text{by permutability}),
\end{align*}
so
$$
G(G(y,v'),r) = G(G(z',s), r).
$$
By the strict monotonicity of $G$ in the first variable, we obtain
$
G(y,v')=G(z',s)
$
and thus
$y\bullet y' = z'\bullet z$.
\tl

(iii) By the solvability condition [S2], there exists $x\in J$ such that, with $G(x_0,r)~=~x$, we have both 
$$
x\bullet x= G(x,r)\in J\quad\text{and}\quad (x\bullet x)\bullet x= G(G(x,r),r) \in J.
$$  

(iv) If $x\bullet y < z$, then $y\bullet x=G(y,r) < z\in J$ by commutativity,  (\ref{commut 1}), and the definition of $\bullet$. Applying [S1], we get $G(w, r) = z$ for some $w\in J$. Using again (\ref{commut 1}),  we obtain $x\bullet w = z$.
\tl

(v) We first show that the sequence $(x_y^n)$ defined by (a) and (b) is strictly increasing. We proceed by induction. Since $x<y$ by definition, we get  from (\ref{commut 1}) and (\ref{commut 2})
$$
x=G(x_0,r)<G(x_0,v)=y,
$$
with the function $G$ strictly monotonic in its second variable. In the sequel, we suppose that 
$G$ is strictly decreasing in its second variable; so,
\begin{equation}\label{v_smaller_than_r}
v<r\,.
\end{equation}
 The proof is similar in the other case.  The following equalities hold by the  definitions of $x_y^1$, $x_y^2$ and commutativity: 
\begin{equation*}
y \bullet x_y^1 = y\bullet x= G(y,r)= x\bullet y= x\bullet x_y^2 =x_y^2 \bullet x= G(x_y^2,r).
\end{equation*}
From $G(y,r)= G(x_y^2,r)$, we get $x_y^2 = y$ and $x_y^1 < x_y^2$.
Assuming that $x_y^{n-1} < x_y^{n}$, we get $y\bullet x_y^{n}= x\bullet x_y^{n+1}$ by the definition of the term $x_y^{n+1}$ in Condition (v) (b) of Lemma \ref{Holder},
and by commutativity
$$
 x_y^{n}\bullet y= G(x_y^{n},v) =  x_y^{n+1} \bullet x=G(x_y^{n+1},r),
$$
yielding $G(x_y^{n},v)  = G(x_y^{n+1},r)$. Since $v<r$ and $G$ is decreasing in its second variable
$$
G(x_y^{n+1}, v) > G(x_y^{n+1},r) = G(x_y^{n},v),
$$
and so 
$$
x_y^n < x_y^{n+1}
$$
because $G$ is strictly increasing in its first variable. By induction, the sequence $(x_y^n)$ is strictly increasing.

Suppose that the set $N(x,z;y)$ of Condition (v) is not finite. Thus, the point $z$ is an upper bound of the sequence $(x_y^n)$. Because this sequence is increasing and  bounded above, it necessarily converges. Without loss of generality, we can assume that we have in fact $\lim_{n\to\infty} x_y^n = z$. Since 
$$
y\bullet x_y^{n-1} = x \bullet x_y^n < x\bullet z
$$
for all $n\in\Na$, the solvability Condition (iv) implies that there is some $z'\in J$ such that $y\bullet z'= x\bullet z$, with necessarily $z'< z$. There must be some $m\in\Na$ such that $z'<x_y^m< z$. We obtain thus
\begin{equation*}
x\bullet z= y\bullet z'
< y\bullet x_y^m = x\bullet x_y^{m+1}
\end{equation*}
and so  $z< x_y^{m+1}$, in contradiction with $\lim_{n\to\infty} x_y^n = z$, with $(x_y^n)$  an increasing sequence.  We conclude that the set $N(x,z;y)$ must be finite for all $x, y$ and $z$ in $J$, with $x < y$. We conclude that the Conditions (i)-(v) of Lemma \ref{Holder} are satisfied.
\tl
To prove that $\bullet$ is associative, we take any $x$, $y$ and $z$ in $J$. Using again $G(x_0, r ) = x$, $G(x_0, v) = y$ and $G(x_0,s) = z$, we have
\begin{align*}
x\bullet (y\bullet z)&= G(y\bullet z, r)&&(\text{since $G(x_0,r) = x$})\\
&=G(G(y,s),r)&&(\text{since $G(x_0,s) = z$})\\
&=G(G(y,r),s)&&(\text{by permutability})\\
&=G(x\bullet y,s)&&(\text{since $G(x_0,r) = x$})\\
&=z\bullet (x\bullet y)&&(\text{since $G(x_0,s) = z$})\\
&=(x\bullet y)\bullet z &&(\text{by commutativity}).
\end{align*}
Finally, since for all $x,y\in J$, we have
$$
x\bullet y = G(y,r)= y\bullet x = G(x, v),
$$
it is clear that the operation $\bullet$ is continuous and strictly increasing in both variables. 
\end{proof}
\section*{Main Result}

The theorem below generalizes results of 
 \citet{hossz62a, hossz62b, hossz62c} \citep[cf.~also][]{aczel:66}.
\begin{theorem}\label{permut => Holder} 
{\rm(i)} A solvable code $M:J\times J'\to H$ is quasi permutable if and only if there exists three  continuous functions $m: \{f(y)+ g(r)\st x\in J,\, r\in J'\}\to H$, $f:J\to \Ree$, and $g:J'\to \Ree$, with $m$ and $f$  strictly increasing and $g$ strictly monotonic, such that
\begin{equation}\label{representation k f g}
M(y, r) = m(f(y)+ g(r)).
 \end{equation}
 
{\rm(ii)}  A solvable code  $G:J\times J'\to J$ is a permutable code if and only if, with $f$ and $g$ as above, we have
 \begin{equation}\label{representation f -1 f g}
G(y, r) = f^{-1}(f(y)+ g(r)).
\end{equation}

{\rm(iii)} If a solvable code $G:J\times J\to J$  is a symmetric function---that is, $G(x,y) = G(y,x)$ for all $x,y\in J$--- then  $G$ is  permutable if and only if there exists a strictly increasing and continuous function $f:J\to J$ satisfying 
 \begin{equation}\label{representation h -1 h h}
G(x, y) = f^{-1}(f(x)+ f(y)).
\end{equation}

{\rm(iv)} If the code $G$ in {\rm (\ref{representation f -1 f g})}  is differentiable in both variables, with non vanishing derivatives, then the functions $f$ and $g$  are differentiable. This differentiability result also applies to the code $G$ and the function $f$ in {\rm (\ref{representation h -1 h h})}.
\end{theorem} 

Our argument for establishing (i) and (ii) is essentially the same as that in \citet[p.~271-273]{aczel:66} but,
because our solvability conditions [S1]-[S2] are weaker, it relies on Lemma~\ref{Holder} rather than on the representation in the reals of an ordered Archimedean group \citep[for example, c.f.][]{holder1901}.
\wl
\begin{proof} (i)-(ii) Suppose that the code $M$ of the theorem is permutable with respect to a $F$-comonotonic code $G$. By Lemma \ref{M permut => G permut}, the code $G$ is permutable. Defining the operation $\bullet:J\times J'\to J$ by
\begin{equation}\label{def circ 2}
y\bullet x = G(y, r)\quad\EQ\quad G(x_0, r) = x,
\end{equation}
it follows from Lemma \ref{G satisfies 1-5} that the triple $(J,\bullet,\leq)$ satisfies 
 Conditions (i)-(v) of Lemma~\ref{Holder}, with the operation $\bullet$ associative and continuously increasing in both variable. Accordingly, there exists a continuous, strictly increasing function $f:J\to J$ such that
 \begin{equation}\label{Holder eq}
 f(y\bullet x ) = f(y)+f(x).
 \end{equation}
 Defining the strictly monotonic function $\psi: J'\to J$ by
 $$
\psi(s) = G(x_0, s), 
 $$
 we get from (\ref{def circ 2}) and (\ref{Holder eq}),
 $$
 f(y\bullet x) = f(G(y,r)) = f(y\bullet G(x_0,r)) = f(y) + f(\psi(r)),
 $$
 and thus
$$
 G(y,r) = f^{-1}( f(y)+f(\psi(r))),
 $$
 or with with $g = f\circ\psi$,
 \begin{equation}\label{G(y,r)= f_1}
 G(y,r) = f^{-1}(f(y) + g(r)).
\end{equation}
(Notice that  $f(y)+ g(r) \in J$.) Because $G$ is $F$-comonotonic with $M$, and  $F$ maps $H$ onto $J$, we obtain
\begin{gather}\nonumber
M(y, r) = F^{-1}(G(y,r))= (F^{-1}\circ f^{-1})(f(y)+ g(r) ),\\
\noalign{\hspace{-.5cm}or, with $m= F^{-1}\circ f^{-1}$,}
\label{M(y,r)= kf(}
M(y, r) = m(f(y)+ g(r))\qquad\quad(y\in J; r\in J';f(y)+ g(r) \in J).
 \end{gather}
 It is clear that the functions $f$ and $g$ in (\ref{G(y,r)= f_1}) and the functions $m$, $f$ and $g$ in (\ref{M(y,r)= kf(}) are continuous, with the required monotonicity properties. This proves the necessity part of~(i). The sufficiency is straightforward. 
\tl

(ii) This was established in passing:  cf.~Eq.~(\ref{G(y,r)= f_1}).
\tl

(iii) From (ii), we get by the symmetry of $G$
\begin{gather*}
G(x,y) =  f^{-1}(f(x)+ g(y)) = G(y,x) =  f^{-1}(f(y)+ g(x))\\
\noalign{yielding}
f(x)-g(x) = f(y)-g(y) = K
\end{gather*}
for some constant $K$ and all $x,y\in J$. We have thus $g(x) = f(x) - K$ for all $x\in J$. Since $g^{-1}(t) = f^{-1}(t+ K)$, we obtain
$$
g^{-1}\left(g(x) + g(y)\right) = f^{-1}\left(g(x) + g(y) + K\right) =  f^{-1}\left(f(x) + g(y)\right) = G(x,y).
$$ 
Defining $h = g$, we obtain (\ref{representation h -1 h h}).
\tl

 (iv) If the code $G$ in (\ref{representation f -1 f g}) is differentiable with non vanishing derivatives, then, for every $r\in J$, the inverse $G^{-1}_r$ of $G$ in the first variable is differentiable.  From (\ref{representation f -1 f g}), we get 
  $ f(x) = G^{-1}_r (x) + g(r)$
 with $ G^{-1}_r (x) = y$. 
So, $f$ is differentiable, and since, from (\ref{representation f -1 f g}) again,
 $$
 g(r) = f\left(G(y,r)\right) - f(y) 
 $$
with $f$ differentiable and $G$ differentiable in the second variable, $g$ is also differentiable. 
The differentiability of $f$ in (\ref{representation h -1 h h}) is immediate.
\end{proof}

We mention in passing a simple uniqueness result concerning our basic representation equation~(\ref{representation f -1 f g}).

\begin{lemma}\label{uniqueness} Suppose that the representation $G(y,r) = f^{-1}(f(y) + g(r))$ of Theorem {\rm \ref{permut => Holder}(ii)} holds for some code $G$, with $f$ and $g$ satisfying the stated continuity and monotonicity conditions. Then we also have 
$G(y,r) = (f^*)^{-1}(f^*(y) + g^*(r))$ for some  continuous functions $f^*$ and $g^*$, respectively co-monotonic with $f$ and $g$, if and only if $f^* = \xi f + \hh$ and $g^*=\xi g$, for some constants $\xi > 0$ and $\hh$.
\end{lemma}

{\sc Proof.} (Necessity.) Suppose that 
$$
(f^*)^{-1}(f^*(y) + g^*(r))= f^{-1}(f(y) + g(r)).
$$
Then, with $z= f(y)$ and $s = g(r)$ and applying $f^*$ on both sides, we get
\begin{equation}\label{Pexi <=>}
(f^*\circ f^{-1})(z) + (g^*\circ g^{-1})(s) = (f^*\circ f^{-1})(z + s),
\end{equation}
a Pexider equation. It is clear that $(f^*\circ f^{-1})$ and $(g^*\circ g^{-1})$ are strictly increasing and continuous and that (\ref{Pexi <=>}) is defined on an open connected subset of $\Ree_+^2$. 
Accordingly  \citep[c.f.][and footnote 1]{aczel:05}, with $h=(f^*\circ f^{-1})$  and $k = m = (g^*\circ g^{-1})$,  we get $(f^*\circ f^{-1})(z)= \xi z + \hh$ and $(g^*\circ g^{-1})(s) = \xi s$, $\xi  > 0$, and so $f^*(y) = \xi f(y) + \hh$ and $g^*(r) = \xi g(r)$.
\tl
 
(Sufficiency.) If $f^* = \xi f + \hh$ and $g^*=\xi g$, with $\xi> 0$, then 
\begin{align*}
(f^*)^{-1}(f^*(y) + g^*(r))&= f^{-1}\left(\frac{f^*(y) + g^*(r) - \hh}\xi\right)\\
&=  f^{-1}\left(\frac{\xi f(y) + \hh + \xi g(r) - \hh}\xi\right)\\
&= f^{-1}(f(y) + g(r)).
\end{align*}

\vspace{-1cm}

\EOP


\begin{thebibliography}{15}
\providecommand{\natexlab}[1]{#1}
\providecommand{\url}[1]{\texttt{#1}}
\expandafter\ifx\csname urlstyle\endcsname\relax
  \providecommand{\doi}[1]{doi: #1}\else
  \providecommand{\doi}{doi: \begingroup \urlstyle{rm}\Url}\fi

\bibitem[Acz{\'e}l(2005)]{aczel:05}
J.~Acz{\'e}l.
\newblock {Utility of extension of functional equations---when possible}.
\newblock \emph{Journal of Mathematical Psychology}, 49:\penalty0 445--449,
  2005.

\bibitem[Acz{\'e}l(1966)]{aczel:66}
J.~Acz{\'e}l.
\newblock \emph{{Lectures on Functional Equations and their Applications}}.
\newblock Academic Press, New York and San Diego, 1966.

\bibitem[Acz{\'e}l(1987)]{aczel:87}
J.~Acz{\'e}l.
\newblock \emph{{A short course on functional equations based on recent
  applications to the social and behavioral sciences}}.
\newblock Reidel/Kluwer, Dordrecht and Boston, 1987.


\bibitem[Falmagne(1975)]{falma75a}
J.-Cl. Falmagne.
\newblock A set of independent axioms for positive {H}\"older systems.
\newblock \emph{Philosophy of Science}, 42\penalty0 (2):\penalty0 137--151,
  1975.

\bibitem[H\"older(1901)]{holder1901}
O.~H\"older.
\newblock Die axiome der quantitat und die lehre von mass.
\newblock \emph{Berichte uber die Verhandlungen der Koniglichen, Sachsischen Gesellschaft der Wissenschaften zu Leipzig, Mathematische-Physysische Classe }, 53:\penalty0 1--64,
  1901.

\bibitem[Hossz\'u(1962{\natexlab{a}})]{hossz62a}
M.~Hossz\'u.
\newblock Note on commutable mappings.
\newblock \emph{Publ.~Math.~Debrecen}, 9:\penalty0 105--106,
  1962{\natexlab{a}}.

\bibitem[Hossz\'u(1962{\natexlab{b}})]{hossz62b}
M.~Hossz\'u.
\newblock N\'eh\'any line\'aris f\"uggv\'enyegyenletr\"ol.
\newblock \emph{Mat.~Lapok}, 13:\penalty0 202, 1962{\natexlab{b}}.

\bibitem[Hossz\'u(1962{\natexlab{c}})]{hossz62c}
M.~Hossz\'u.
\newblock Algebrai rendszereken\'ertelmezett f\"uggv\'enyegyenletek, i.
  algebrai m\'odszerek a f\"uggv\'enyegyenletek elm\'elet\'eben.
\newblock \emph{Magyar Tud.~Acad.~Mat.~Fiz.~Oszt.~K\:ozl}, 12:\penalty0
  303--315, 1962{\natexlab{c}}.
  
  \bibitem[Krantz et al.(1971)]{krantzetal71}
D.H. Krantz, R.D. Luce, P. Suppes, and A. Tversky
\newblock \emph{{Foundations of Measurement, Vol. I}}.
\newblock Academic Press, York, 1971.

  \bibitem[Luce et al.(1990)]{luceetal90}
R.D. Luce, D.H. Krantz, P. Suppes, and A. Tversky
\newblock \emph{{Foundations of Measurement, Vol. III}}.
\newblock Academic Press, York, 1990.
  
  \bibitem[Luce and Marley(1969)]{lucemarley69}
R.D. Luce and  A.A.J. Marley
\newblock {Extensive measurement when concatenation is restricted and maximal elements may exist}.
\newblock {in S. Morgenbesser, P. Suppes, and M.G. White (Eds.},
\newblock {\emph{Philosophy, science and method: essays in honor of Ernest Nagel}},
\newblock {New York: St. Martins Press,  \penalty0 235--249, 1969.}
  
  \bibitem[Suppes et al.(1989)]{suppesetal89}
P. Suppes, D.H. Krantz, R.D. Luce,  and A. Tversky
\newblock \emph{{Foundations of Measurement, Vol. II}}.
\newblock Academic Press, York, 1989.
\end{thebibliography}

\end{document}